\documentclass[11pt]{article}     
\usepackage{amsmath,amssymb,amsthm,ascmac}
\usepackage{mathrsfs}
\usepackage[dvips]{graphicx}
\theoremstyle{plain}
\newtheorem{thm}{Theorem}[section]

\newtheorem{lem}{Lemma}[section]

\newtheorem{cor}{Corollary}[section]

\def\r0{{\mathscr{R}_0}}

\theoremstyle{definition}

\numberwithin{equation}{section}

\begin{document}
\begin{center}
   {\bf {\Large Backward bifurcation of a disease-severity-structured epidemic model with treatment}}  
\end{center}

\vspace{3mm}    

\begin{center}
   {\bf Hiromu Gion, \quad Yasuhisa Saito$^*
\footnote{corresponding author, whose e-mail address: ysaito@riko.shimane-u.ac.jp}$\\
   Department of Mathematics, Shimane University, Japan}  
\end{center}

\begin{abstract}
This paper presents 
a disease-severity-structured epidemic model
with treatment necessary only to  severe infective individuals
to discuss 
the effect of the treatment capacity on the
disease transmission.
It is shown 
that a 
backward bifurcation occurs  in the basic 
reproduction number $\mathscr{R}_0$,  
where a stable 
endemic equilibrium
co-exists with a stable 
disease-free equilibrium
when $\mathscr{R}_0 <1$,
if the capacity is relatively small. 
This epidemiological implication is 
that, when there is not enough capacity for treatment, the requirement $\mathscr{R}_0 <1$ 
is not sufficient for effective disease control and disease
outbreak can happen to a high endemic level even though $\mathscr{R}_0 <1$.
\end{abstract}

\begin{center}
\end{center}

\section{Introduction}

Treatment is an important method to decrease the spread of diseases such as measles, tuberculosis
and flu (see, for example \cite{thieme}). In classical disease transmission models, the treatment rate of infective individuals is assumed to
be proportional  to their number. 
However, it is natural to think that there is some capacity 
for the treatment, including limited beds in hospitals, or an insufficient supply of medicine. 
Such a limited capacity, 
in \cite{wang},  was considered as
\begin{align*}
\frac{dS}{dt} &= A-\sigma SI -\mu S\\
\frac{d I}{dt} &=\sigma SI -(\mu+\rho+\varepsilon)I -T(I)\\
\frac{dR}{dt} &=T(I) + \rho I-\mu R
\end{align*} 
where  $T(\cdot)$ is the treatment rate defined as
the 
function of $I$:
\begin{eqnarray*}
T(I)=
 \begin{cases}
   r I, \enskip  & I < C_I  \\
   r C_I, \enskip & I \geq C_I
  \end{cases}
\end{eqnarray*}
with the per capita treatment rate $r$ and the capacity $C_I$. Here, 
$S(t)$, $I(t)$, and $R(t)$ denote the numbers of susceptible, infective,  and recovered individuals 
at time $t$, respectively. 
$A$ is the recruitment rate of the population, $\mu$ the per capita natural death rate of the population, 
$\varepsilon$ the per capita disease-related death rate, 
$\rho$
the per capita natural recovery rate of infective individuals,  $\sigma$ the disease 
transmission coefficients of infective individuals.
Characterizing the dynamics of disease transmission models often requires
 the basic 
reproduction number $\mathscr{R}_0$, the average number of new cases that would
 be generated by a typical infected individual introduced into a completely susceptible 
population. 
In general, the phenomenon \textit{forward} bifurcation is observed, where the disease-free 
equilibrium loses its stability and a stable endemic equilibrium appears as $\mathscr{R}_0$
 increases through one. 
\cite{wang} 
figured out, however,  \textit{backward} bifurcations occur, 
 where a stable endemic equilibrium co-exists with a stable disease-free equilibrium 
when $\mathscr{R}_0 <1$ (as illustrated in Fig. 3 in \cite{wang}), 
due to the low treatment capacity. The fact implies that the requirement $\mathscr{R}_0 <1$ 
is not sufficient for effective disease control and disease outbreak can happen to a high endemic 
level even though $\mathscr{R}_0 <1$. 

In case of  
less-lethal diseases (not to mention non-lethal diseases),  
we can  assume $\varepsilon =0$ when the disease-related death rate 
is negligibly small compared with the natural death rate.
Besides, 
for such diseases,
infective  individuals do not always need treatment in that 
they can recover by themselves in mild case.
To discuss 
the effect of the treatment capacity 
on the 
transmission of non-lethal  or less-lethal diseases,
in this paper, 
we consider 
the following 
disease-severity-structured epidemic model,
 which is more
realistic than  
the above mentioned model
if  $\varepsilon =0$:
\begin{equation}
\begin{aligned}
 \frac{dS}{dt} &= A-\sigma_mSI_m -\sigma_sS[I_s]_{C_I}^{+} 
-\mu S\\
\frac{d I_m}{dt} &=\sigma_mSI_m+ \sigma_sS[I_s]_{C_I}^{+} 
-(\mu+\rho+\beta)I_m\\
\frac{dI_s}{dt}&=\beta I_m -T(I_s)-\mu I_s\\
 \frac{dR}{dt} &= T(I_s)+\rho I_m-\mu R.
\end{aligned} \label{sec1:eq1}
\end{equation}
Here, new unknown functions  $I_s(t)$ and $I_m(t)$ denote the number of severe infective 
individuals who need treatment, and  
the number of non-severe (that is, mild) infective individuals who do not need it, respectively.
Also, new parameters $\sigma_m$ and $\sigma_s$ are the disease 
transmission coefficients of mild and severe infective individuals, respectively,  
$\beta$ the severity coefficient of mild infective individuals. Furthermore, 
$$[I_s]_{C_I}^{+}=\max\{0, I_s-C_I\},$$
by which $\sigma_sS[I_s]_{C_I}^{+}$ means the transmission rate of  severe infective individuals
exceeding the capacity. 
We assume all the parameters to be positive constants and the initial data given as
\begin{equation*}
S(0)>0,   
\quad  
I_m(0) \geq 0, \quad I_s(0) \geq 0, 
\quad I_m(0)+I_s(0) >0, \quad R(0) \geq 0. 
\end{equation*}
As a result, backward bifurcation occurs for \eqref{sec1:eq1}, 
which leads to the same scenario as mentioned above, 
if  the treatment capacity $C_I$ is relatively small. 
Details of our results and their proof are found in the next section and  Section 3. 
We summarize our findings in Section 4.

\section{Basic reproduction number and equilibria}

To analyze \eqref{sec1:eq1}, 
we only 
focus on the three dimensional ODEs:
\begin{equation}
\begin{aligned}
 \frac{dS}{dt} &= A-\sigma_mSI_m -\sigma_sS[I_s]_{C_I}^{+} 
-\mu S\\
\frac{d I_m}{dt} &=\sigma_mSI_m+ \sigma_sS[I_s]_{C_I}^{+} 
-(\mu+\rho+\beta)I_m\\
\frac{dI_s}{dt}&=\beta I_m -T(I_s)-\mu I_s
\end{aligned} \label{sec1:eq3}
\end{equation}
since the first three equations in \eqref{sec1:eq1} are independent of the variable $R$. 
For \eqref{sec1:eq3},  uniqueness of the solutions is ensured by the Lipschitz continuity
of its right-hand side 
although 
the right-hand side 
is not differentiable because of $[I_s]_{C_I}^{+}$ and $T(I_s)$.

Disease-free equilibrium is required to derive the basic reproduction number $\r0$. 
Obviously, the model \eqref{sec1:eq3}  has always a disease-free equilibrium $E_0$ expressed as 
$$
 (S, I_m, I_s) = \left(\frac{A}{\mu}, 0, 0\right).
$$
According to the concept of \textit{next generation matrix}, 
let
$x={}^t(I_m, I_s, S)$
and write \eqref{sec1:eq3} as
\[
\dfrac{dx}{dt}=\mathscr{F}(x)-\mathscr{V}(x),  
\]
where $\mathscr{F}$ is the rate of production term of new infection 
and $-\mathscr{V}$ otherwise. 
From the viewpoint of local behavior in completely susceptible population,  
the first principal matrices 
(defined as $\mathcal{F}, \mathcal{V}$ below) of the Jacobian matrix of $\mathscr{F}$ and $\mathscr{V}$ at disease-free equilibrium play a key role of defining $\r0$ 
(see \cite{Diekmann}, \cite{van}). 
Then the disease-free equilibrium 
is 
as
$(I_m, I_s, S) = \left(0, 0, A/\mu \right)$
and the first principal matrices $\mathcal{F}, \mathcal{V}$ are given as
\begin{equation*}
\mathcal{F}
=\begin{pmatrix}
\frac{A\sigma_m}{\mu}&0\\[3mm]
      0&0
\end{pmatrix},  
\qquad 
\mathcal{V}
=\begin{pmatrix} \mu+\beta + \rho &0\\[3mm]
 -\beta & \mu+r
\end{pmatrix}.
\end{equation*}
Hence, the basic reproduction number for \eqref{sec1:eq3} is 
\begin{equation*}
\r0=
\|\mathcal{F}\mathcal{V}^{-1}\|
=\frac{A\sigma_m}{\mu(\mu+\beta+\rho)}, 
\end{equation*}
where $\|\mathcal{M}\|$ represents the spectral radius of the matrix $\mathcal{M}$.

An endemic equilibrium of \eqref{sec1:eq3} satisfies 
\begin{align*}
  & A-\sigma_mSI_m -\sigma_sS[I_s]_{C_I}^{+} -\mu S =0 \\ 
 & \sigma_mSI_m+ \sigma_sS[I_s]_{C_I}^{+} -(\mu+\rho+\beta)I_m =0 \\
 & \beta I_m -T(I_s)-\mu I_s =0, 
\end{align*} 
which becomes
\begin{equation}
\begin{aligned}
  & A-\sigma_mSI_m  -\mu S =0 \\ 
 & \sigma_mSI_m -(\mu+\rho+\beta)I_m =0 \\
 & \beta I_m -(\mu +r) I_s =0 
\end{aligned} \label{sec1:eq4}
\end{equation}
if $0 < I_s \leq C_I$, while 
\begin{equation}
\begin{aligned}
  & A-\sigma_mSI_m  -\sigma_s S(I_s - C_I)-\mu S =0 \\ 
 & \sigma_mSI_m +\sigma_s S(I_s - C_I)-(\mu+\rho+\beta)I_m =0 \\
 & \beta I_m - \mu I_s -r C_I =0 
\end{aligned} \label{sec1:eq5}
\end{equation}
if $C_I < I_s$. When $0 < I_s \leq C_I$ and $\r0 >1$, \eqref{sec1:eq4} admits a unique
positive solution $E^* = (S^*, I_m^*, I_s^*)$:
\begin{equation*}
 S^* = \frac{\mu+ \beta + \rho}{\sigma_m}, 
\quad 
 I_m^* = \frac{\mu (\r0 -1)}{\sigma_m}, 
\quad 
I_s^* = \frac{\mu \beta (\r0 -1)}{\sigma_m (\mu +r)}.
\end{equation*}
Clearly, $E^*$ is an endemic equilibrium of \eqref{sec1:eq3} if and only if 
\begin{equation*}
 1 < \r0 \leq 1+ \frac{\sigma_m (\mu + r)}{\mu \beta} C_I.
\end{equation*}

In order to  obtain positive solutions of \eqref{sec1:eq5}, 
solving \eqref{sec1:eq5} in terms of $I_m$, $I_s$ with $S$ and substituting the result into 
the first equation, we have
\begin{equation*}
 S^2 - \frac{\mu + \beta + \rho}{\sigma_m}\left(\r0 - p +q \right) S 
 + \frac{(\mu + \beta + \rho)^2}{\sigma^2_m} q \r0 =0, 
\end{equation*}
where 
\begin{equation*}
 p=\frac{(\mu +r) \sigma_m \sigma_s C_I}{\mu (\mu \sigma_m + \beta \sigma_s)}, 
\qquad q= \frac{\mu \sigma_m}{\mu \sigma_m + \beta \sigma_s}. 
\end{equation*}
Note that $p >0$ and $0 < q <1$. 
When  
$C_I < I_s$, 
\eqref{sec1:eq5} admits possible two positive solutions 
$E_1^* = (S_1^*, I_{m_1}^*, I_{s_1}^*)$, 
$E_2^* = (S_2^*, I_{m_2}^*, I_{s_2}^*)$ where
\begin{align*}
  S_1^* &= \frac{\mu + \beta +\rho}{2\sigma_m} \left\{\r0 - p +q - \sqrt{(\r0 - p -q)^2-4pq} \right\},\\
 S_2^* &= \frac{\mu + \beta +\rho}{2\sigma_m} \left\{\r0 - p +q+ \sqrt{(\r0 - p -q)^2-4pq} \right\}, \\
    I_{m_i}^* & 
 = \frac{\mu \r0}{\sigma_m}- \frac{\mu S_i^*}{\mu + \beta + \rho}, 
\quad 
 I_{s_i}^* = \frac{\beta \r0}{\sigma_m}- \frac{\beta S_i^*}{\mu + \beta + \rho} - \frac{r C_I}{\mu},  \qquad i= 1, 2.
\end{align*}
We see that $I_{m_i}^* >0$ and $I_{s_i}^* > C_I$ are equivalent to 
$$
S < \frac{\mu + \beta + \rho}{\sigma_m} \left(\r0 -\frac{\sigma_m (\mu + r)C_I}{\mu \beta} \right),
$$
which implies that \eqref{sec1:eq3} does not have any endemic equilibria 
if $\r0 \leq \sigma_m (\mu + r)C_I/(\mu \beta)$.
From the above, we have the following:

\begin{lem}
 \eqref{sec1:eq3} has always a unique disease-free equilibrium $E_0$, and has an endemic equilibrium 
$E^*$ if and only if  
$$
1 < \r0 \leq 1+ \frac{\sigma_m (\mu + r)}{\mu \beta} C_I.
$$
Moreover, \eqref{sec1:eq3} does not have any endemic equilibria  in the region $C_I < I_s$ if 
\begin{equation*}
 C_I \geq \frac{\beta A}{(\mu + r)(\mu + r + \rho)}. 
\end{equation*}
\end{lem}

\vspace{5mm}

For convenience, let 
\begin{equation*}
 a=   \frac{\mu + \beta + \rho}{\sigma_m} \left(\r0 -\frac{\sigma_m (\mu + r)C_I}{\mu \beta} \right). 
\end{equation*}
In order to find endemic equilibria  in the region $C_I < I_s$, 
by Lemma 2.1 we should
consider the case 
\begin{equation}
C_I < \frac{\beta A}{(\mu + r)(\mu + r + \rho)},  \label{sec1:eq71}
\end{equation}
which is equivalent to 
$
 \r0 > \sigma_m (\mu + r)C_I/(\mu \beta).
$
To figure out $S_1^*$, $S_2^*$ feasible for endemic equilibria, define 
\begin{equation}
f(S)= S^2 - \frac{\mu + \beta + \rho}{\sigma_m}\left(\r0 - p +q \right) S 
 + \frac{(\mu + \beta + \rho)^2}{\sigma^2_m} q \r0. \label{sec1:eq70}
\end{equation}
When \eqref{sec1:eq71} holds, this quadratic function axis is positive since 
\begin{equation}
\frac{\sigma_m (\mu + r)C_I}{\mu \beta} > p > p-q.  \label{sec1:eq72}
\end{equation}
Then we see that 
$E_1^*$ is an endemic equilibrium of \eqref{sec1:eq3} but not $E_2^*$ if and only if  
$f(a) \leq 0$ holds as the axis of $f(S)$ is less than $a$, 
or
$f(a) < 0$ holds as the axis of $f(S)$ is greater than or equal to $a$,
which equivalently implies 
\begin{equation}
\begin{aligned}
 & \r0 \geq 1+ \frac{\sigma_m (\mu + r)C_I}{\mu \beta} \enskip \text{and} 
\enskip \r0 > -p +q + \frac{2\sigma_m (\mu + r)C_I}{\mu \beta} \\ 
\text{or} & \\
 & \r0 > 1+ \frac{\sigma_m (\mu + r)C_I}{\mu \beta} \enskip \text{and} 
\enskip \r0 \leq -p +q + \frac{2\sigma_m (\mu + r)C_I}{\mu \beta} 
\end{aligned} \label{sec1:eq721}
\end{equation}
holds, respectively. 
Since we obtain
\begin{equation}
 1+ \frac{\sigma_m (\mu + r)C_I}{\mu \beta} -\left( -p +q + \frac{2\sigma_m (\mu + r)C_I}{\mu \beta} \right)
= -\frac{(\mu+r) \sigma_m^2}{\beta (\mu \sigma_m + \beta \sigma_s)}
\left(C_I-\frac{\beta^2 \sigma_s}{(\mu + r) \sigma_m^2} \right),  \label{sec1:eq723}
\end{equation}
\eqref{sec1:eq721} is equivalent to  
\begin{equation*}
 \r0 > 1+ \frac{\sigma_m (\mu + r)C_I}{\mu \beta} \qquad \text{if} \quad C_I \geq \frac{\beta^2 \sigma_s}{(\mu + r) \sigma_m^2}
\end{equation*}
while 
\begin{equation*}
 \r0 \geq 1+ \frac{\sigma_m (\mu + r)C_I}{\mu \beta} \qquad \text{else.} 
\end{equation*}
Similarly, 
both of $E_1^*$ and $E_2^*$ are endemic equilibria of \eqref{sec1:eq3} if and only if
\begin{equation}
\begin{aligned}
 & f(a) > 0,\\
 &  (\r0 - p -q)^2-4pq \geq 0, \\
 &  
\frac{\mu + \beta + \rho}{2\sigma_m}\left(\r0 - p +q \right)  < a.
\end{aligned} \label{sec1:eq8}
\end{equation}
Clearly, the first and third conditions of \eqref{sec1:eq8} are equivalent to 
\begin{equation}
 \r0 < 1+ \frac{\sigma_m (\mu + r)C_I}{\mu \beta} \label{sec1:eq9}
\end{equation}
and 
\begin{equation}
 \r0 > -p +q + \frac{2\sigma_m (\mu + r)C_I}{\mu \beta}, \label{sec1:eq10}
\end{equation}
respectively.  
Also, by \eqref{sec1:eq72} and \eqref{sec1:eq10}, the second one of \eqref{sec1:eq8} 
is replaced with
\begin{equation}
  \r0 \geq \left(\sqrt{p} + \sqrt{q}\right)^2. \label{sec1:eq11}
\end{equation}  
Thus, it follows from the facts \eqref{sec1:eq723} and 
\begin{align*}
 & \left(\sqrt{p} + \sqrt{q}\right)^2 - \left( -p +q + \frac{2\sigma_m (\mu + r)C_I}{\mu \beta} \right)
 = -\frac{2(\mu + r) \sigma_m^2 \sqrt{C_I}}{\beta (\mu \sigma_m + \beta \sigma_s)}\left(\sqrt{C_I}
-\sqrt{\frac{\beta^2 \sigma_s}{(\mu + r) \sigma_m^2}} \right),  \\
& 1+ \frac{\sigma_m (\mu + r)C_I}{\mu \beta}- \left(\sqrt{p} + \sqrt{q}\right)^2 
 = \frac{(\mu+r) \sigma_m^2}{\beta (\mu \sigma_m + \beta \sigma_s)}\left(\sqrt{C_I}
-\sqrt{\frac{\beta^2 \sigma_s}{(\mu + r) \sigma_m^2}} \right)^2 
\end{align*}
that a common range where \eqref{sec1:eq71}, \eqref{sec1:eq9}, \eqref{sec1:eq10}, and \eqref{sec1:eq11} hold
 is
\begin{equation}
  \left(\sqrt{p} + \sqrt{q}\right)^2 \le \r0 < 1+ \frac{\sigma_m (\mu + r)C_I}{\mu \beta} \label{sec1:eq12}
\end{equation}
if 
$
 C_I < \beta^2 \sigma_s/\{(\mu + r) \sigma_m^2\} 
$
while no the common range exists else. Furthermore, 
by 
some 
tedious calculation, we see 
that
$
 \sqrt{p} + \sqrt{q} <1,
$
$
 \sqrt{p} + \sqrt{q} =1
$
are equivalent to 
$$
 C_I 
<\frac{\mu^2}{\mu + r}\left(\sqrt{\frac{\beta}{\mu \sigma_m}+\frac{1}{\sigma_s}}-\sqrt{\frac{1}{\sigma_s}}\right)^2, 
\quad C_I =\frac{\mu^2}{\mu + r}\left(\sqrt{\frac{\beta}{\mu \sigma_m}+\frac{1}{\sigma_s}}-\sqrt{\frac{1}{\sigma_s}}\right)^2, 
$$
respectively. 
Here, 
$$
 \frac{\mu^2}{\mu + r}\left(\sqrt{\frac{\beta}{\mu \sigma_m}+\frac{1}{\sigma_s}}-\sqrt{\frac{1}{\sigma_s}}\right)^2
= \frac{\beta^2}{(\mu + r)\sigma_m^2 \left(\sqrt{\frac{\beta}{\mu \sigma_m}+ \frac{1}{\sigma_s}}+ \sqrt{\frac{1}{\sigma_s}} \right)^2}
< \frac{\beta^2 \sigma_s}{(\mu + r) \sigma_m^2}.
$$
Hence, we have the following:
\begin{thm}
 Suppose \eqref{sec1:eq71} and 
$$
 \frac{\beta^2 \sigma_s}{(\mu + r) \sigma_m^2} \le C_I
$$
hold. 
Then $E^*$ is a unique endemic equilibrium of \eqref{sec1:eq3} 
if $1 < \r0 \leq 1+ \frac{\sigma_m (\mu + r)}{\mu \beta} C_I$, and $E_1^*$ is
a unique endemic equilibrium 
if  $\r0 > 1+ \frac{\sigma_m (\mu + r)}{\mu \beta} C_I$.
\end{thm}

\begin{thm}
 Suppose \eqref{sec1:eq71} and
$$
 \frac{\mu^2}{\mu + r}\left(\sqrt{\frac{\beta}{\mu \sigma_m}+\frac{1}{\sigma_s}}-\sqrt{\frac{1}{\sigma_s}}\right)^2  < C_I < \frac{\beta^2 \sigma_s}{(\mu + r) \sigma_m^2}
$$
hold. Then $E^*$ is a unique endemic equilibrium of \eqref{sec1:eq3} if  $1 < \r0 < (\sqrt{p}+ \sqrt{q})^2$, 
all endemic equilibria $E^*$, $E_1^*$, and $E_2^*$ exist if  
$(\sqrt{p}+ \sqrt{q})^2 \leq \r0 < 1+ \frac{\sigma_m (\mu + r)}{\mu \beta} C_I$, 
$E_2^*$ does not exist 
but $E^*$ and $E_1^*$ exist 
if $\r0 = 1+ \frac{\sigma_m (\mu + r)}{\mu \beta} C_I$, 
and 
 $E_1^*$ is
a unique endemic equilibrium if  $\r0 > 1+ \frac{\sigma_m (\mu + r)}{\mu \beta} C_I$.
\end{thm}

\begin{thm}
 Suppose \eqref{sec1:eq71} and
$$
 C_I \leq \frac{\mu^2}{\mu + r}\left(\sqrt{\frac{\beta}{\mu \sigma_m}+\frac{1}{\sigma_s}}-\sqrt{\frac{1}{\sigma_s}}\right)^2
$$
hold. Then $E^*$ does not exist as endemic equilibrium for \eqref{sec1:eq3} but $E_1^*$ and $E_2^*$ exist 
if  $(\sqrt{p}+ \sqrt{q})^2 \leq \r0 \leq 1 $, 
all endemic equilibria $E^*$, $E_1^*$, and $E_2^*$ exist if  
$1 < \r0 < 1+ \frac{\sigma_m (\mu + r)}{\mu \beta} C_I$, 
$E_2^*$ does not exist 
but $E^*$ and $E_1^*$ exist 
if $\r0 = 1+ \frac{\sigma_m (\mu + r)}{\mu \beta} C_I$, 
and 
 $E_1^*$ is
a unique endemic equilibrium if  $\r0 > 1+ \frac{\sigma_m (\mu + r)}{\mu \beta} C_I$.
\end{thm}

Note that a backward bifurcation with endemic equilibria when $\r0 <1$ is very interesting in applications.
We present the following corollary, a consequence of Theorems 2.1--2.3, to 
give a necessary and sufficient condition
for such a backward bifurcation to occur.

\begin{cor}
 \eqref{sec1:eq3} has a backward bifurcation with endemic equilibria when $\r0 <1$ if and only if \eqref{sec1:eq71} and
$$
 C_I \leq \frac{\mu^2}{\mu + r}\left(\sqrt{\frac{\beta}{\mu \sigma_m}+\frac{1}{\sigma_s}}-\sqrt{\frac{1}{\sigma_s}}\right)^2.
$$
\end{cor}

\section{Stability of the equilibria}

We have the following results on stability for 
all the equilibria of \eqref{sec1:eq3}.

\begin{thm}
 $E_0$ is asymptotically stable if $\r0 <1$, but unstable if $\r0 >1$. 
\end{thm}

\begin{proof}
 This theorem is a simple consequence of Theorem 2 of \cite{van}.
\end{proof}

\begin{thm}
$E^*$ is asymptotically stable  
if $1 < \r0 < 1+ \frac{\sigma_m (\mu + r)}{\mu \beta} C_I$. 
$E_1^*$ is
asymptotically stable whenever it exists and does not shrink to $E_2^*$, while 
$E_2^*$ is unstable
whenever it exists and does not shrink to $E_1^*$.
\end{thm}

\begin{proof}
 We analyze the eigenvalues of the Jacobian matrices of \eqref{sec1:eq3} at the equilibria, to which 
Lemma A.28 in \cite{thieme} and Rough-Hurwiz criteria (see, for example \cite{britton}) are applied. 
First, 
we have 
the Jacobian matrix 
at $E^*$:
\begin{equation*}
\mathcal{J}(E^*) = \left(
\begin{array}{ccc}
-\frac{A}{S^*} & -\sigma_m S^* & 0  \\
\frac{A}{S^*}-\mu & 0 &  0 \\
 0  & \beta &  -\mu - r  
\end{array}
\right)
\end{equation*}
and then 
consider the characteristic equation 
$
 \det(\lambda \mathcal{I} -\mathcal{J}(E^*))=0
$
with an identity matrix $\mathcal{I}$, 
which is given as 
$$
 (\lambda+\mu + r)\left\{\lambda^2 + \frac{A}{S^*}\lambda +\mu(\mu + \beta + \rho)(\r0 -1) \right\}=0.
$$ 
It is clear that all eigenvalues 
of $\mathcal{J}(E^*)$ have negative real parts 
when $1 < \r0 < 1+ \frac{\sigma_m (\mu + r)}{\mu \beta} C_I$, which implies that $E^*$ is asymptotically stable.

For the Jacobian matrix 
at $E_2^*$:
\begin{equation*}
\mathcal{J}(E_2^*) = \left(
\begin{array}{ccc}
-\frac{A}{S_2^*} & -\sigma_m S_2^* & -\sigma_s S_2^*  \\
\frac{A}{S_2^*}-\mu & \sigma_m S_2^* - \mu-\beta-\rho &  \sigma_s S_2^* \\
 0  & \beta &  -\mu  
\end{array}
\right),
\end{equation*}
the characteristic equation $\det(\lambda \mathcal{I} -\mathcal{J}(E_2^*))=0$ is given as  
$$
 (\lambda+\mu)\left\{\lambda^2 + \left(\frac{A}{S_2^*}-\sigma_m S_2^* + \mu + \beta + \rho \right)\lambda 
- \frac{(\mu + \beta + \rho)(\mu \sigma_m + \beta \sigma_s)\sqrt{D_1}}{\sigma_m}\right\} =0,
$$ 
where $D_1 =(\r0-p-q)^2-4pq$ and we used relations between roots and coefficients for $S_1^*$ and $S_2^*$.
Note that $D >0$ since we now consider the case $S_1^* \ne S_2^*$. 
Then it is clear that one eigenvalue 
of $\mathcal{J}(E_2^*)$ is a positive real number, which implies that $E_2^*$ is unstable.

In order to consider the stability of $E_1^*$, we similarly have the Jacobian matrix 
at $E_1^*$:
\begin{equation*}
\mathcal{J}(E_1^*) = \left(
\begin{array}{ccc}
-\frac{A}{S_1^*} & -\sigma_m S_1^* & -\sigma_s S_1^*  \\
\frac{A}{S_1^*}-\mu & \sigma_m S_1^* - \mu-\beta-\rho &  \sigma_s S_1^* \\
 0  & \beta &  -\mu  
\end{array}
\right).
\end{equation*}
and  
the characteristic equation 
$$
 (\lambda+\mu)\left\{\lambda^2 + \left(\frac{A}{S_1^*}-\sigma_m S_1^* + \mu + \beta + \rho \right)\lambda 
+ \frac{(\mu + \beta + \rho)(\mu \sigma_m + \beta \sigma_s)\sqrt{D_1}}{\sigma_m}\right\} =0.
$$ 
To conclude
$E_1^*$ is asymptotically stable,  we only have to show that 
$$
 \frac{A}{S_1^*}-\sigma_m S_1^* + \mu + \beta + \rho>0,
$$
which is proven separately in the following two cases.
Recall that
\begin{equation}
 0 < S_1^* < \frac{\mu + \beta + \rho}{\sigma_m} \left(\r0 -\frac{\sigma_m (\mu + r)C_I}{\mu \beta} \right)=a. \label{sec3:eq0}
\end{equation}

\vspace{3mm}

\noindent
(i) 
The case 
where
$
\r0 < 1 + \sigma_m(\mu + r)C_I/(\mu \beta) 
$
holds. 
By \eqref{sec3:eq0} we easily see that
\begin{align*}
 \frac{A}{S_1^*}-\sigma_m S_1^* + \mu + \beta + \rho &> 
\frac{A}{S_1^*}-(\mu + \beta +\rho)\left(\r0 -\frac{\sigma_m (\mu + r)C_I}{\mu \beta} \right) 
 + \mu + \beta + \rho \\
   &=\frac{A}{S_1^*}-(\mu + \beta +\rho)\left(\r0 -1- \frac{\sigma_m (\mu + r)C_I}{\mu \beta} \right) >0.
\end{align*}

\vspace{2mm}

\noindent
(ii) 
When 
$
\r0 \geq 1 + \sigma_m(\mu + r)C_I/(\mu \beta) 
$
holds,
it follows that
\begin{equation}
f(a) 
 = -\frac{(\mu + \beta + \rho)^2 q  (\mu + r)C_I}{\sigma_m \mu \beta}\left(\r0 -1-\frac{\sigma_m (\mu + r)C_I}{\mu \beta}  \right) \leq 0, \label{sec3:eq1}
\end{equation}
where $f$ is defined by \eqref{sec1:eq70}.
We obtain
\begin{align*}
 & \frac{A}{S_1^*}-\sigma_m S_1^* + \mu + \beta + \rho \\
 & \quad = \frac{\mu+\beta + \rho}{S_1^*}
\left\{-\left(\r0 -1-p+q\right)S_1^* +\frac{(\mu + \beta +\rho) \r0}{\sigma_m}\left(q+\frac{\mu}{\mu + \beta +\rho}\right) \right\} 
\end{align*} 
since $f(S_1^*) =0$. Define
$$
 g(S)= -\left(\r0 -1-p+q\right)S +\frac{(\mu + \beta +\rho) \r0}{\sigma_m}\left(q+\frac{\mu}{\mu + \beta +\rho}\right).
$$
Then the proof will be completed if $g(S_1^*) >0$, 
equivalent to
\begin{equation}
 0 < S_1^* <\frac{(\mu + \beta +\rho) \r0}{\sigma_m \left(\r0 -1-p+q\right)}\left(q+\frac{\mu}{\mu + \beta +\rho}\right) \label{sec3:eq2}
\end{equation}
since $\r0 -1-p+q>0$ 
by \eqref{sec1:eq72}. 

Now we consider a quadratic function $h$ given as 
$$
 h(x)= x^2 -\left(1 +\frac{\sigma_m(\mu + r)C_I}{\mu \beta} + p + \frac{\mu}{\mu + \beta + \rho}\right)x 
+(1+p-q)\frac{\sigma_m(\mu + r)C_I}{\mu \beta}. 
$$
By some tedious calculation, we see that $h(\r0)  \leq 0$ if and only if
$$
 a
\leq \frac{(\mu + \beta +\rho) \r0}{\sigma_m \left(\r0 -1-p+q\right)}\left(q+\frac{\mu}{\mu + \beta +\rho}\right), 
$$ 
which, together with \eqref{sec3:eq1}, implies that \eqref{sec3:eq2} holds. 
Clearly, 
$h(x)$ has the two zeros $\alpha_1$, $\alpha_2$:
\begin{align*}
 \alpha_1 &=  \frac{1}{2}\left(1 +\frac{\sigma_m(\mu + r)C_I}{\mu \beta} + p + \frac{\mu}{\mu + \beta + \rho}- \sqrt{D_2} \right), \\
 \alpha_2 &= \frac{1}{2}\left(1 +\frac{\sigma_m(\mu + r)C_I}{\mu \beta} + p + \frac{\mu}{\mu + \beta + \rho}+ \sqrt{D_2} \right),
\end{align*}
where 
$$
 D_2= \left(p + \frac{\mu}{\mu + \beta + \rho} - 1 - \frac{\sigma_m(\mu + r)C_I}{\mu \beta} \right)^2 
+ \frac{ (1+p-q) \sigma_m(\mu + r)C_I}{\beta p(\mu + \beta + \rho)} >0.
$$
Note that
$
\alpha_1 < 1 + \sigma_m(\mu + r)C_I/(\mu \beta) < \alpha_2.
$
Hence, only consideration in the case
\begin{equation}
\r0 > \alpha_2 \label{sec3:eq3}
\end{equation} 
remains to complete the proof.
Let
$$
 b= \frac{(\mu + \beta +\rho) \r0 q}{\sigma_m \left(\r0 -1-p+q\right)}.
$$
Then we have 
\begin{equation}
 b <
\frac{(\mu + \beta +\rho) \r0}{\sigma_m \left(\r0 -1-p+q\right)}\left(q+\frac{\mu}{\mu + \beta +\rho}\right) < a, \label{sec3:eq4}
\end{equation}
where the first inequality is clearly verified and the second one is by the fact that
the case \eqref{sec3:eq3} ensures $h(\r0) > 0$ as discussed above. Furthermore, we obtain
\begin{equation}
 f(b) = \frac{q(q-1)}{\left(\r0 - 1-p+q \right)^2}\left(\r0 - 1- \frac{\sigma_m(\mu + r)C_I}{\mu \beta} \right) <0
\label{sec3:eq5}
\end{equation}
since $q <1$ and \eqref{sec3:eq3}. It follows from \eqref{sec3:eq1}, \eqref{sec3:eq4}, \eqref{sec3:eq5}, and 
the convexity of the function $f$ that 
$$
 f\left(\frac{(\mu + \beta +\rho) \r0}{\sigma_m \left(\r0 -1-p+q\right)}\left(q+\frac{\mu}{\mu + \beta +\rho}\right)\right) < 0,
$$
which implies that \eqref{sec3:eq2} holds. The proof is then complete.
\end{proof}

\begin{figure}[h]
\begin{center}
\begin{minipage}[t]{115mm}
\includegraphics[width=112mm]{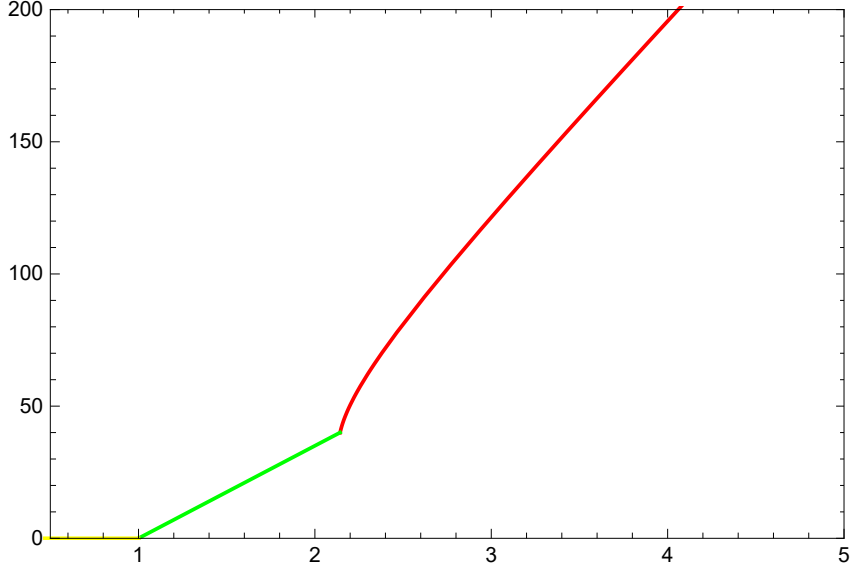}
\vspace{-3mm}
\caption{
A bifurcation diagram 
with $C_I=40$
that satisfies $\frac{\beta^2 \sigma_s}{(\mu + r) \sigma_m^2} \le C_I$, where the vertical axis shows $\r0$ and
the horizontal one the value $I_s$ in equilibrium. The bifurcation
from the disease-free equlibrium at $\r0 =1$ is forward and \eqref{sec1:eq3} has a unique endemic equilibrium for $\r0 >1$.}
\end{minipage}
\end{center}
\end{figure}

\begin{figure}[h]
\begin{center}
\begin{minipage}[t]{115mm}
\includegraphics[width=112mm]{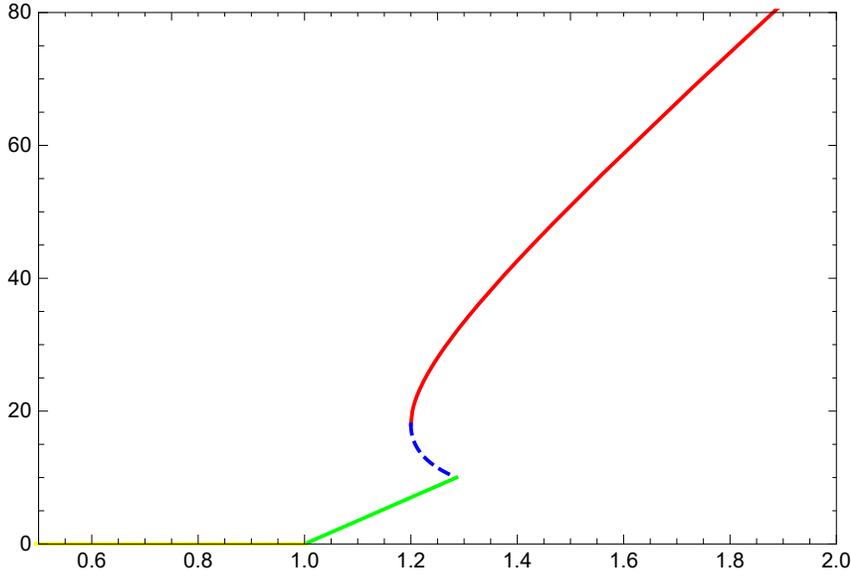}
\vspace{-3mm}
\caption{
A bifurcation diagram 
with $C_I=10$
that satisfies $
 \frac{\mu^2}{\mu + r}\left(\sqrt{\frac{\beta}{\mu \sigma_m}+\frac{1}{\sigma_s}}-\sqrt{\frac{1}{\sigma_s}}\right)^2  < C_I < \frac{\beta^2 \sigma_s}{(\mu + r) \sigma_m^2}
$, where the vertical and
horizontal axes
are the same as Fig. 1. Dashed line represents the unstable equilibrium $E_2^*$. The bifurcation
 at $\r0 =1$ is forward and there is a backward bifurcation from an endemic equilibrium
at $\r0 = 
1+ \frac{\sigma_m (\mu+r)C_I}{\mu \beta}= 
1.286$, which leads to the existence of multiple endemic equilibria. }
\end{minipage}
\end{center}
\end{figure}

\begin{figure}[h]
\begin{center}
\begin{minipage}[t]{115mm}
\includegraphics[width=112mm]{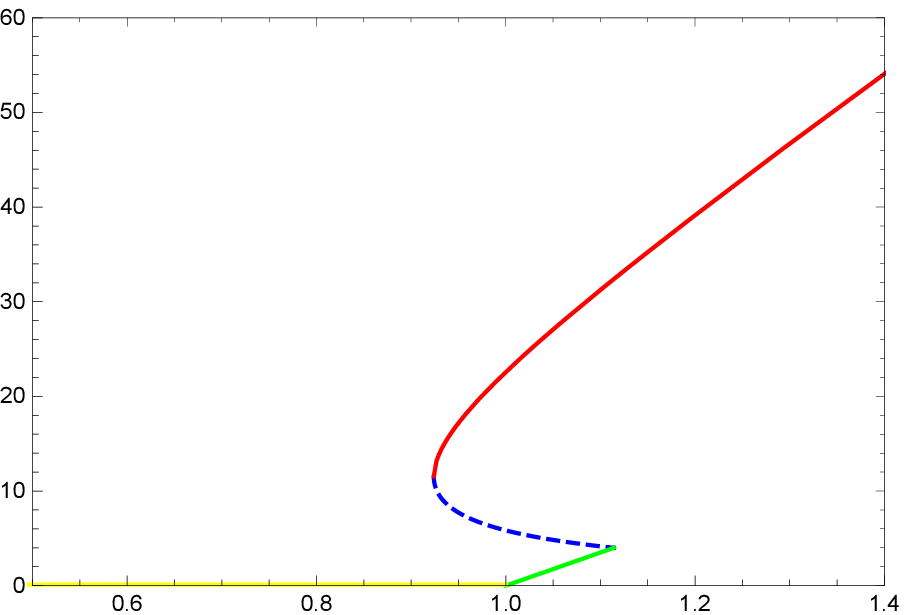}
\vspace{-3mm}
\caption{
A bifurcation diagram 
with  $C_I=4$ 
that satisfies $C_I \leq 
 \frac{\mu^2}{\mu + r}\left(\sqrt{\frac{\beta}{\mu \sigma_m}+\frac{1}{\sigma_s}}-\sqrt{\frac{1}{\sigma_s}}\right)^2
$, where the vertical and
horizontal axes are the same as Fig. 1. 
Dashed line represents the unstable equilibrium $E_2^*$. 
The graph shows 
a backward bifurcation with endemic equilibria when $\r0 <1$.}
\end{minipage}
\end{center}
\end{figure}

$E^*$ exists even if $\r0 = 1+ \sigma_m (\mu + r)C_I/(\mu \beta)$ as mentioned in the previous section.
For the case, 
the stability of $E^*$
is not determined by 
the same method as in the proof of Theorem 3.2
since the right-hand side of \eqref{sec1:eq3} is not differentiable at $E^*$.   
Such situations 
where the equal sign just holds, however, 
cannot be expected to be found in nature and 
can be assumed to 
be neglected 
without loss of biological generality. 
Theorems 3.1 and 3.2 then completely give local properties of all equilibrium solutions
for \eqref{sec1:eq3}.
Typical bifurcation diagrams are
illustrated in Figs. 1--3 in correspondence with three sizes of $C_I$, where other parameters are 
given as $\sigma_m= \sigma_s= 0.01, \mu=0.8, \rho=1.0, r=0.8$, and $\beta=0.7$.
In particular, we present the following corollary to clarify the stability of endemic equilibria for the
backward-bifurcation case mentioned in Corollary 2.1.

\begin{cor}
 Suppose \eqref{sec1:eq71} and
\begin{equation}
 C_I \leq \frac{\mu^2}{\mu + r}\left(\sqrt{\frac{\beta}{\mu \sigma_m}+\frac{1}{\sigma_s}}-\sqrt{\frac{1}{\sigma_s}}\right)^2 \label{sec3:eq10}
\end{equation}
hold.
Then $E_1^*$ is asymptotically stable while $E_2^*$ is  unstable
if  $(\sqrt{p}+ \sqrt{q})^2 \leq \r0 \leq 1 $, 
both $E^*$ and $E_1^*$ are asymptotically stable while $E_2^*$ is unstable if  
$1 < \r0 < 1+ \frac{\sigma_m (\mu + r)}{\mu \beta} C_I$, and 
 $E_1^*$ is
asymptotically stable if  $\r0 > 1+ \frac{\sigma_m (\mu + r)}{\mu \beta} C_I$.
\end{cor}

\section{Concluding remarks}

In this paper, we have proposed a 
disease-severity-structured epidemic model
with treatment necessary only to  severe infective individuals, 
which is 
higher dimensional and also more realistic  than \cite{wang} for non-lethal  or less-lethal diseases,
to discuss 
the effect of the treatment capacity on the
disease transmission. 
Our bifurcation analysis reveals local properties of all the equilibrium solutions 
to fully mathematically obtain 
bifurcation diagrams for any situation.
We have shown in Corollary 2.1 that backward bifurcations occur because of the insufficient capacity
for treatment, which generalizes \cite{wang} in the non-lethal  disease case or in the less-lethal disease 
case where the disease-related death rate is too small to be neglected. 
Once a backward bifurcation occurs, as shown in Corollary 3.1, 
a stable endemic equilibrium co-exists with a stable disease-free equilibrium 
when $\mathscr{R}_0 <1$. This leads to the same scenario as \cite{wang}, that is, 
the requirement $\mathscr{R}_0 <1$ 
is not sufficient for effective disease control and disease outbreak can happen to a high endemic 
level even though $\mathscr{R}_0 <1$.

By Corollary 3.1, no backward bifurcations occur if $C_I$ is so large that \eqref{sec1:eq71} does not hold.
When $C_I$ is not so large and satisfy \eqref{sec1:eq71}, however, it follows from Theorems 2.1-2.3 that  
a backward bifurcation can
occur depending on either or both of $\sigma_m$ and $\sigma_s$.
A novel aspect of our results is a clear formulation
representing the effect of the treatment capacity  
by the necessary and sufficient condition \eqref{sec3:eq10} 
derived for the occurrence of backward bifurcations.
This means that we understand in an explicit way how each parameter, as well as $C_I$, plays a role of  
preventing or promoting the backward-bifurcation  scenario. 
Let us fix $C_I$ such that
\eqref{sec3:eq10} does not hold. Then no backward bifurcations
occur. 
Even if $\sigma_m$ increases, \eqref{sec3:eq10}  has not held forever  since 
its right-hand side 
is a decreasing function in $\sigma_m$, which implies that $\sigma_m$ plays a role of preventing 
backward bifurcation from occurring. 
On the other hand, 
from a similarly simple observation, 
we can see that  $\sigma_s$ 
plays a role of promoting the backward-bifurcation  scenario.

We have not shown global properties of all solutions for \eqref{sec1:eq3}. This is because our model has 
a higher dimension than that of \cite{wang}
and then cannot have the same method as in \cite{wang}, that is, using the
theory of planar dynamical system, applied to it.
In order to carry out
global analysis for \eqref{sec1:eq3},  
 appropriate Lyapunov functions should be constructed when an endemic equilibrium uniquely exists, 
while much more sophisticated mathematics may be required when multiple endemic equilibria coexists, which will be left for future work.   
Besides that, 
as with treatment, vaccination is important to decrease the spread of diseases,
and the number of vaccines is also practically limited. Based on the work of this paper,
we can consider an epidemic model with capacities of vaccination as well as treatment 
to discuss 
the effect of both these capacities on the
disease transmission, which  will be presented on another occasion.

\vspace{5mm}

\noindent
{\bf Acknowledgments}

We thank Professor Hiromi Seno,
Tohoku University, for his helpful comments on mathematical modeling for the infection force of severe infective 
indivisuals.
This work is partially supported by JSPS KAKENHI Grant Numbers 20K03750 (the second author).



\begin{thebibliography}{9}

\bibitem{britton}
N. F. Britton, \textit{Essential Mathematical Biology}, Springer, London, 2003.

\bibitem{Diekmann}
O. Diekmann, J.A.P. Heesterbeek, and J.A.J. Metz, 
On the definition and the computation of the basic reproduction ratio $R_0$ in models for infectious diseases in heterogeneous populations, \textit{J. Math, Biol.} 28 (1990) 365-382. 

\bibitem{thieme}
H. R. Thieme, 
\textit{Mathematics in Population Biology}, 
Princeton University Press, Princeton, 2003.


\bibitem{van}
P. van den Driessche and J. Watmough, 
Reproduction numbers and sub-threshold endemic equilibria for compartmental models of disease transmission, 
\textit{Math. Biosci.} 180 (2002) 29-48.

\bibitem{wang}
W. Wang, 
Backward bifurcation of an epidemic model with treatment, 
\textit{Math. Biosci.} 201 (2006) 58-71.


\end{thebibliography}
\end{document}